\documentclass[aps,prl,superscriptaddress,reprint,nofootinbib]{revtex4-1}
\usepackage{amsthm,amsmath,physics,mathtools,amssymb}%mathrsfs,amsfonts,

\usepackage{hyperref}
\hypersetup{
    colorlinks=true,
    linkcolor=blue,
    filecolor=magenta,      
    urlcolor=cyan,
}

\usepackage{cleveref}

\usepackage{tabularx}

\usepackage{CJKutf8}

\usepackage{enumitem}

\theoremstyle{definition}

\newtheorem{theorem}{Theorem}

\newtheorem{corollary}{Corollary}[theorem]

\usepackage{graphicx}
\graphicspath{ {images/} }

\DeclarePairedDelimiterX{\inp}[2]{\langle}{\rangle}{#1, #2}

\newcommand\id{\leavevmode\hbox{\small1\kern-3.3pt\normalsize1}}

\newcommand{\ablue}[1]{\textcolor{black}{#1}}

\usepackage[normalem]{ulem}

\begin{document}

\begin{CJK*}{UTF8}{gbsn}
\title{Causal order as a resource for quantum communication}
\author{Ding Jia (贾丁)}
\email{ding.jia@uwaterloo.ca}
\affiliation{Department of Applied Mathematics, University of Waterloo, Waterloo, Ontario, N2L 3G1, Canada}
\affiliation{Perimeter Institute for Theoretical Physics, Waterloo, Ontario, N2L 2Y5, Canada}
\author{Fabio Costa}
\affiliation{Centre for Engineered Quantum Systems, School of Mathematics and Physics, University of Queensland, QLD 4072 Australia}
%\date{}

\begin{abstract}
In theories of communication, it is usually presumed that the involved parties perform actions in a fixed causal order. However, practical and fundamental reasons can induce uncertainties in the causal order. Here we show that a maximal uncertainty in the causal order forbids asymptotic quantum communication, while still enabling the noisy transfer of classical information. Therefore causal order, like shared entanglement, is an additional resource for communication. The result is formulated within an asymptotic setting for processes with no fixed causal order, which sets a basis for a quantum information theory in general quantum causal structures.
\end{abstract}

\maketitle
\end{CJK*}

\section{Introduction}

%Definite causal structure as a resource. Ordinary (definite causal structure) communication theory investigates the resources of channels and initial correlations. The major task is to find out the capacities for different resources. Less noisy channels and stronger initial correlations both help to improve the capacities.

One of the basic questions in quantum information theory is to characterise the resources necessary for the reliable transmission of quantum information \cite{wilde2017quantum}: a sender encodes a quantum state in a system and a receiver has to retrieve it, without a prior knowledge of what the state might be \cite{shor1995scheme}. A typical protocol can involve the physical transfer of the system, allowing it to undergo some time evolution possibly including noise. The essential resource is then how well the evolution preserves the initial state. A different method is teleportation \cite{bennett1993teleporting}, where the resources are entanglement and classical communication. The quality of the communication resource is typically measured as the rate of reliable transmitted qubits per use of the resource, in the limit of many independent uses.

Typical communication protocols presume that the involved parties' actions take place in a fixed causal order, with the sender's always preceding the receiver's. More general situations are possible: The parties might both act on a quantum particle that is exchanged between the two, but without knowing to whom the particle goes first. For multiple runs of the protocol, the particle could go one direction or the other randomly, according to some probability. It is natural to ask whether communication is at all possible without a fixed causal order and if causal order should itself be considered as a resource for communication tasks. This can be relevant in scenarios of distributed quantum computation, where separated units have to communicate in order to perform a joint operation, but unknown delays in the network might produce uncertainty in the order in which the units are queried \cite{lamporttime1978}. It is also relevant for foundational questions, such as the informational properties of processes in scenarios where quantum-gravity fluctuations generate uncertainty in causal relations \cite{butterfield2001spacetime, hardy2005probability, hardy2007towards}.

Here we find that, in the asymptotic limit, a communication protocol where the order  between two parties is completely unknown allows the transfer of classical information in either direction (although with limited efficiency), but \emph{not of quantum information}. In particular, we prove that, when the causal order between two parties is completely uncertain (with equal probability for both orders) %\sout{and share no entanglement},
the asymptotic quantum communication capacity vanishes in both directions. 

%Our result is couched in the process matrix framework , which includes more general scenarios where the causal order can be indefinite, rather than simply classically unknown. ... 

%Result: for $p\le 1/2$, no qc for whatever channels. This implies that for a classical mixture of channels of two directions, there is positive qc in at most one direction.

\section{Single-shot Process matrices}

We consider a general communication protocol where, in an individual run, each party receives a quantum system, which might contain information sent by another party or shared correlations, and then sends away a system in which they encoded the desired information. Each party can perform an arbitrary local operation on their system, namely they can let it interact with a local ancilla in some controlled way.  Crucially, the parties have no access to any background causal structure, thus the time of their operations is not set in advance and it might vary probabilistically for different runs of the protocol.

\ablue{Situations of this type are conveniently modelled within the process matrix framework \cite{oreshkov2012quantum,araujo2015witnessing,oreshkov2016causal}, which generalises standard states and channels to scenarios with no background causal structure.} %\sout{The process matrix framework \cite{oreshkov2012quantum,araujo2015witnessing,oreshkov2016causal} is introduced to study indefinite causal structure in the context of operational formulations of quantum theory.} 
We review the framework as formulated in \cite{araujo2017purification} through the notion of higher order maps \cite{perinotti2017causal, bisio2018axiomatic}, which turns out to be convenient for the current study of communication protocols. %\sout{, where a process matrix acts as a higher order map that sends local operations to a new quantum channel.}

\begin{figure}
    \centering
    \includegraphics[width=.5\textwidth]{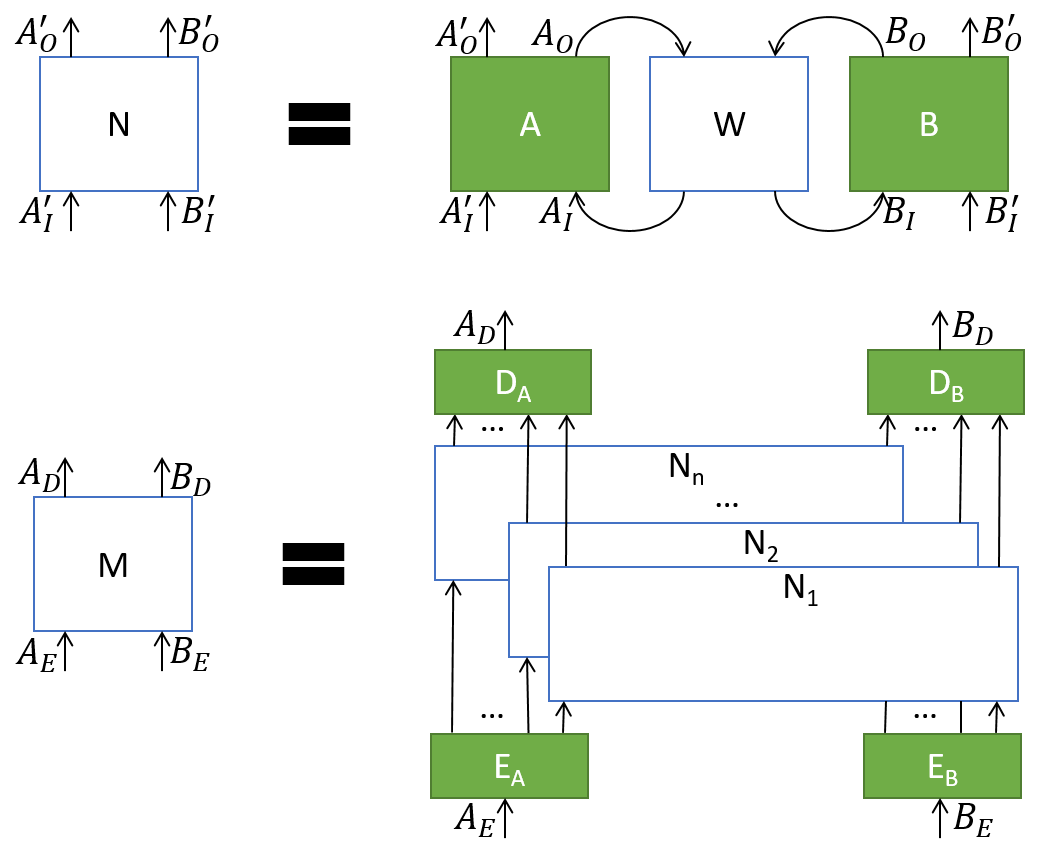}
    \caption{Top: A process $W$ composed with local channels $A$ and $B$ to form a shared channel $N$. Bottom: $n$ such shared channels $N_1,N_2,\cdots,N_n$ composed with local encodings $E_A, E_B$ and decodings $D_A, D_B$ to form a shared channel $M$.}
    \label{fig:bp}
\end{figure}

Bipartite processes defined as higher-order maps are illustrated in the top of \cref{fig:bp}. Let $A$ be a completely positive trace preserving (CPTP) map, i.e., a quantum channel, with the input systems $A_I,A_I'$ and output systems $A_O,A_O'$, \ablue{where $A_I$, $A_O$ represent the system that Alice receives and respectively send back to the process, while $A_I'$, $A_O'$ represent her local ancilla before and after the interaction}. Similarly, let $B$ be a quantum channel with inputs $B_I,B_I'$ and outputs $B_O,B_O'$. \ablue{A processes $W$ is defined as a linear map acting on $A_I,A_O,B_I,B_O$, with the requirement that any pair of such channels $A$ and $B$ is transformed into a new channel $N=W(A,B)$, with inputs $A_I',B_I'$ and outputs $A_O',B_O'$.}
%\sout{The processes $W$ are defined as the most general linear maps acting on the systems $A_I,A_O,B_I,B_O$ to take arbitrary such channels $A$ and $B$ to new channels $N=W(A,B)$ with input systems $A_I',B_I'$ and output systems $A_O',B_O'$.}

%\dcomment{This paragraph rewritten}
Quantum channels and processes can be represented as matrices through the Choi isomorphism~\cite{choi1975completely}, which takes a completely positive map $M:\mathcal{L}(\mathcal{H}^{a_1})\rightarrow \mathcal{L}(\mathcal{H}^{a_2})$ to its positive semidefinite ``Choi matrix''
%\begin{align}\label{eq:choi}
%(M\otimes \id) \sum_{i,j=1}^{d_{a_1}}\ketbra{ii}{jj}\in \mathcal{L}(\mathcal{H}^{a_2}\otimes\mathcal{H}^{a_1}),
%\end{align}
\begin{align}\label{eq:choi}
\sum_{i,j=1}^{d_{a_1}} M(\ketbra{i}{j})\otimes \ketbra{i}{j}\in \mathcal{L}(\mathcal{H}^{a_2}\otimes\mathcal{H}^{a_1}),
\end{align}
where $\{\ket{i}\}$ is an orthonormal basis of $\mathcal{H}^{a_1}$. In the following, we refer to the Choi matrix of a map $M$ using the same letter. In this representation, the defining property of a process $W$, that it  maps channels to channels, is captured by the following \cite{araujo2015witnessing}: %\sout{ at the process matrix level}:
\begin{align} \label{positive}
W\ge& 0, 
\\
\Tr W=& d_O,
\\
_{B_IB_O}W =& _{A_OB_IB_O}W,
\\ \label{marginalB}
_{A_IA_O}W=&_{B_OA_IA_O}W,
\\ \label{noloops}
W=&_{A_O}W+_{B_O}W-_{A_OB_O}W,
%W=&L_V W.\label{eq:pmp1}
\end{align}
 where $d_O$ is the product dimension of the output systems and we use the trace-and-replace notation $_X M := \frac{\id^X}{d^X}\otimes \Tr_X M$, where $d_X$ is the dimension of system $X$.
Any matrix $W$ obeying conditions \eqref{positive}--\eqref{noloops} is called a process matrix. It defines a process, whose action on channels $A$ and $B$ is represented using Choi matrices as $N(A,B)=W*A*B=W*(A\otimes B)$. Here we use the ``link product'' \cite{chiribella2009theoretical}:
%\begin{align}\label{eq:lp}
%M*N:=\Tr_{M\cap N}[(I_{N \backslash M}\otimes M^{T_{M \cap N}})(N\otimes I_{M\backslash N})].
%\end{align}
\begin{align}\label{eq:lp}
M*N:=\Tr_{s}[M^{T_{s}} N],
\end{align}
where $s$ is the system that is the joint support of $M$ and $N$, and $T_s$ is the partial transpose on $s$. It is understood that the operators act as the identity outside of its original support.

%\begin{align}
%L_A W:=&W-_{B_IB_O}W+_{A_OB_IB_O}W,
%\\
%L_B W:=&W-_{A_IA_O}W+_{B_OA_IA_O}W,
%\\
%L_{AB} W:=&_{A_O}W+_{B_O}W-_{A_OB_O}W,
%\end{align}
%where $_X M := \frac{\id^X}{d^X}\otimes \Tr_X M$, and $d_X$ is the dimension of system $X$. For arbitrary $W$, because $L_V W=W$ and $L_A,L_B, L_{AB}$ mutually commute,
%\begin{align}
%W=L_A W=L_B W=L_{AB} W. \label{eq:wcp}
%\end{align}

It is useful to consider causally-ordered processes, i.e., processes that cannot transmit information in certain directions. We use $W^{A\prec B}$ (respectively $W^{A\succ B}$) to denote a process that cannot be used to signal from $B$ to $A$ (respectively $A$ to $B$). It holds that \cite{gutoski2007toward, chiribella2009theoretical}
\begin{align}
W^{A\prec B}=&_{B_O}W^{A\prec B},\label{eq:watb}
\\
W^{A\succ B}=&_{A_O}W^{A\succ B}.\label{eq:wbta}
\end{align}
%\fcomment{The trace-and-replace notation should be introduced earlier, where it's first used.} 
Incidentally, the processes generalize quantum channels. For instance, one can verify that a quantum channel from $A$ to $B$ is a special case of $W^{A\prec B}$ with the systems $A_I$ and $B_O$ set to be the one-dimensional trivial system. General causally-ordered processes represent \emph{channels with memory} \cite{kretschmann2005quantum}.

Here we are interested in more general processes, $W^{AB}=pW^{A\prec B}+(1-p)W^{A\succ B}$, where Alice might come before Bob with probability $p$ and Bob before Alice with probability $1-p$. Such processes are called \emph{causally separable}, and it is known that more general situations, where the causal order is indefinite, are possible too \cite{oreshkov2012quantum, chiribella2013quantum, baumeler14, Branciard2016, abbott2016, giacomini2016indefinite, zych2017bell}. In this work, however, we are mostly concerned with definite, albeit possibly unknown, causal order.

\section{The asymptotic setting}\label{sec:as}

Measures of communication capacity are typically defined as the optimal rate of transmitted information per use of the resource, in the limit of infinite uses~\cite{shannon1948mathematical, wilde2017quantum}. Recall that, in the standard asymptotic setting for channel communication, $n\rightarrow\infty$ copies $N^{\otimes n}$ of the channel $N$ are sandwiched between a joint encoding channel $E$ and a joint decoding channel $D$. 
%This setting is based on the assumption that one party (the encoding party) definitely causally precedes the other (the decoding party). When this assumption is relaxed so that there is no definite causal order between the two parties, we have a generalized asymptotic setting. Communication may in principle be conducted either from Alice to Bob or from Bob to Alice, so a party can possibly apply both an encoding operation and a decoding operation.

To generalise this notion to processes where the causal order is not fixed, we need to clarify in what ways the parties can use multiple copies of a process. Each copy $W_i$ of the process is associated with input-output spaces $A^i_I, A^i_O$, which can be accessed by letting them interact with ancillary systems ${A^i_I}', {A^i_O}'$ (and similarly for Bob). As the parties have no access to a background causal structure, they do not know in which order different copies will be instantiated. Therefore, they can communicate in no other way than through the process. In other words, $n$ uses of a bipartite process $W$ are described by the the $2n$-partite process $W^{\otimes n}$ where, for $i=1,2,\cdots, n$, each party can only apply independent channels $A_i$, $B_i$. (We can equivalently say that the process $W^{\otimes n}$ can be composed with arbitrary product channels $\bigotimes_{i=1}^n \left(A_i\otimes B_i\right)$.) We can still think that all the $A_i$ channels are controlled by a single agent, Alice, and the $B_i$ ones by Bob, who are restricted to product channels because of the unknown causal order. We will refer to Alice and Bob as ``agents", to distinguish them from ``parties", which we reserve to the individual access points to each copy of the process.
%In doing this, we have to make sure that all potential communication resources are modelled through the shared process, and we are not introducing extra resources in the way the processes are used. For example, if the parties were allowed to share extra entanglement in their ancillae, or to use side communication channels, we would not be able to attribute a successful communication protocol to the process under consideration.

%For this reason, we associate $n$ uses of a bipartite process $W$ to the $2n$-partite process $W^{\otimes n}$ where, for $i=1,2,\cdots, n$, each party can apply independent channels $A_i$, $B_i$ with inputs $A^i_I, {A^i_I}'$ (resp., $B^i_I, {B^i_I}'$) and outputs $A^i_O, {A^i_O}'$ (resp., $B^i_O, {B^i_O}'$). %The asymptotic setting for quantum communication we study is illustrated in the bottom of \cref{fig:bp}. 
Formally, at each iteration the parties convert the process into a channel for the corresponding ancillary systems, $N_i:=W_i*(A_i\otimes B_i)$. The $N_i$'s thus obtained can then be used according to the ordinary asymptotic settings for channels, with the difference that now each agent can both receive and send information. Therefore, the channel $N^{\otimes n}$ can be preceded by encoding channels $E_A, E_B$, which prepare joint states in the spaces ${A^{\otimes n}_I}'$, ${B^{\otimes n}_I}'$, and followed by decoding channels $D_A, D_B$, which transform respectively  ${A^{\otimes n}_O}'$, ${B^{\otimes n}_O}'$ into the final output state. Note that no joint encoding on ${A^{\otimes n}_I}'\otimes {B^{\otimes n}_I}'$ should be allowed, as this could introduce additional entanglement, not modelled in the process, and thus an additional communication resource. Similarly, common decodings on  ${A^{\otimes n}_O}'\otimes {B^{\otimes n}_O}'$ are excluded, as they would allow the parties to exchange information beyond what is enabled by the process. As illustrated in the bottom of \cref{fig:bp}, the setting described generates a shared channel $M$, which can be used to communicate information.

%\fcomment{Feel free to edit this paragraph}

We observe that there are \ablue{also other} %\sout{different}
ways in which agents could use multiple copies $W^{\otimes n}$ of a process, for example by arranging the access to the different parties in some given order or adding extra entanglement. Note, however, that there are constraints on process composition: For example, $W^{\otimes n}$ cannot be used as a bipartite process, where all the $A$'s and the $B$'s each act `simultaneously' as a collective party~\cite{jia2018tensor, guerin2019composition}. The setting introduced above is appropriate for the study of causal order as a communication resource in the asymptotic setting, as it precludes the agents from using additional communication resources, such as entanglement or causal order. A more general study can be based on the one-shot setting, which is beyond the scope of this work.

% Recall that in the standard asymptotic setting for channel communication, $n\rightarrow\infty$ copies $N^{\otimes n}$ of the channel $N$ is sandwiched between a joint encoding channel $E$ and a joint decoding channel $D$. This setting is based on the assumption that one party (the encoding party) definitely causally precedes the other (the decoding party). When this assumption is relaxed so that there is no definite causal order between the two parties, we have a generalized asymptotic setting. Communication may in principle be conducted either from Alice to Bob or from Bob to Alice, so a party can possibly apply both an encoding operation and a decoding operation. 

%Formally, a communication protocol consists of the quantum channels $N_i:=W_i*(A_i\otimes B_i)$ for $i=1,2,\cdots, n$ sandwiched between the encodings $E_A, E_B$ applied at Alice's and Bob's sides, respectively, and decodings $D_A, D_B$ applied at Alice's and Bob's sides, respectively. %The quantum channel $N_i:=W_i*(A_i\otimes B_i)$ is generated by composing the local channels $A_i$ at Alice's side and $B_i$ at Bob's side with the $i$-th process $W_i$.

\section{The quantum communication task}\label{sec:qct}

The standard definition of the quantum communication task for channels through entanglement or subspace transmission \cite{wilde2017quantum} can be generalized to processes~\footnote{The definition through entanglement generation generalized to processes, however, does not capture ``communication''. States are special cases of processes. Entanglement can be generated from them even when they do not allow signalling at all from one party to the other.}, and their capacities agree~\footnote{This can be seen as a consequence of the theorem in the appendix of \cite{jia2018quantum2} and the fact that the capacities agree for channels \cite{barnum2000quantum}.}. Without loss of generality we present the communication task and define the quantum communication capacity for processes through entanglement transmission.

In an entanglement transmission task from Alice to Bob, Alice, in addition to sharing copies of the process $W^{AB}$ with Bob, also shares a preexisting state $\tau$ with a third party Charlie. The goal is for Alice to ``transmit'' her share of the preexisting state to Bob so that in the end Bob and Charlie share a state $\rho$ that is as close to $\tau$ as possible. In the above asymptotic setting, the protocol takes the form
\begin{align}
\rho^{CB_D}=M*\tau^{CA_E}=(N^{\otimes n}*E_A*E_B*D_B)*\tau^{CA_E},
\end{align}
%\begin{align}
%\rho^{CB_D}=M*\sigma^{CA_E}*\tau^{B_E}=\Tr_{A_DB_E}(N^{\otimes n}*E_A*D_A*E_B*D_B*\tau^{B_E})*\sigma^{CA_E},
%\end{align}
where $N^{\otimes n}=\otimes_{i=1}^n N_i$ and $N_i=W_i*(A_i\otimes B_i)$. Without loss of generality the system $A_O'$ and the operation $D_A$ have been taken to be trivial, since they are not accessible to Bob and will eventually be traced out. In addition, $E_B$ is taken to be a state rather than a channel, since even if it were a channel in the beginning, a state needs to be fed into its input to turn the channel into a state by the end of the protocol.

We say there is a $(R,n,\epsilon)$ code for entanglement transmission if for $R=(1/n)\log m$ there exists a protocol with $(A_i,B_i,E_A,E_B,D_B)$ such that for any input state $\sigma^{CA}$ with $\dim{A}=\dim{C}=m$, the fidelity $F(\sigma^{CA},\rho^{CB_D})\ge 1-\epsilon$. A rate $R$ is said to be achievable if there is a sequence of $(R,n,\epsilon_n)$ codes with $\epsilon_n\rightarrow 0$. The \textit{quantum communication capacity} of $W$, $Q(W)$, is the supremum of the achievable rates.

\section{Results}

Suppose two agents Alice and Bob can interact with multiple uses of the process $W$ in the above asymptotic setting. How does the lack of a definite causal order affect their quantum communication capacity? We know that if $W=W^{A\succ B}$, Alice cannot communicate any information to Bob, but what about the general case of a causally separable $W=pW^{A\prec B}+(1-p)W^{A\succ B}$, which becomes $W^{A\succ B}$ only at $p=0$? Here we prove that the quantum capacity actually starts to vanish at the much higher value $p=1/2$, which implies that, for any causally separable $W$ of this form, there is quantum capacity in at most one direction.
\begin{theorem}\label{th:ngbwc}
$W^{AB}=pW^{A\prec B}+(1-p)W^{A\succ B}$ can have positive quantum communication capacity %(for entanglement transmission) 
in the Alice to Bob direction if and only if $p>1/2$. 
\end{theorem}
\begin{proof}
In a general protocol, Alice applies a channel $A_i$ on the $i$-th copy $W_i$ of $W$, and as explained above, ${A_O^i}'$ can be taken to be trivial.
\begin{align}
&W_i*A_i
\\
=&p W^{A\prec B}_i*A_i+(1-p)W^{A\succ B}_i*A_i
\\
=&p W^{A\prec B}_i*A_i+(1-p)[_{A_O^i}W^{A\succ B}_i*A_i]\label{eq:2nd}
\\
=&pW^{A\prec B}_i*A_i+(1-p) \id^{{A_I^i}'}\otimes\Tr_{A_I^iA_O^i}W^{A\succ B}_i\label{eq:3rd}
\\
=&pW^{A\prec B}_i*A_i+(1-p)\id^{{A_I^i}'B_O^i}\otimes \sigma^{B_I^i}\label{eq:4th}
\\
=&[p\Tr_{B_O^i}W^{A\prec B}_i*A_i+(1-p)\id^{{A_I^i}'}\otimes \sigma^{B_I^i}]\otimes \id^{B_O^i},\label{eq:aae}
%\\
%=&[p [_{B_O}W^{A\prec B}_i*A_i]+(1-p)\id^{{A_I^i}'}\otimes \sigma^{B_I}]\otimes \id^{B_O},
\end{align}
where is the density operator defined by $\sigma^{B_I^i}:= \Tr_{A_I^iA_O^iB_O^i}W^{A\succ B}$. \Cref{eq:2nd} holds by \cref{eq:wbta}. \Cref{eq:3rd} holds because $\id^{A_O^i}*A_i=\id^{A_I^i{A_I^i}'}$, which is true for any channel $A_i$. \Cref{eq:4th} holds by \cref{marginalB}. \Cref{eq:aae} holds by \cref{eq:watb}. $\id^{B_O^i}$ in (\ref{eq:aae}) means that whatever that is sent into $B_O^i$ is traced out. This implies that $E_B$ and $B_i$ can be omitted. Their non-trivial part is obtained after tracing out $B_O^i$, which can be absorbed into $D_B$.

The communication resource above can be simulated by a quantum erasure channel \footnote{A quantum erasure channel is defined to send an input state $\rho$ to the output state $p\rho+(1-p)\ketbra{e}$, where the erasure flag state $\ket{e}$ is orthogonal to any possible input state \cite{grassl1997codes}. The idea is that when the erasure occurs the receiver can detect it from the flag state in an orthogonal subspace.}. In a communication protocol, the $\Tr_{B_O}W^{A\prec B}_i*A_i\otimes \id^{B_O^i}$ part of \Cref{eq:aae} is equivalent to the channel $L_i:=\Tr_{B_O^i}W^{A\prec B}_i*A_i$, while the $\id^{{A_I^i}'B_O^i}\otimes \sigma^{B_I^i}$ part is equivalent to a fixed state $\sigma^{B_I^i}$. %$W_i*A_i$ can then be viewed as a channel that with probability $p$ performs $L_i$ and probability $1-p$ outputs a state. The overall protocol consists $n$ such...
$W_i*A_i$ can be simulated by a quantum erasure channel $\rho\rightarrow p\rho+(1-p)\ketbra{e}$: Bob applies the local channel $L_i$ when no erasure occurs, and locally sends $\ketbra{e}$ to $\sigma$ when the erasure occurs. By doing this for all $i$, the whole protocol of communication using $W$ can be simulated by one using the erasure channel. Consequently the quantum capacity of $W^{AB}$ is upper-bounded by that of the quantum erasure channel, which is known to be $Q=\max\{0,2p-1\}$ \cite{bennett1997capacities}. Therefore the capacity of the process to communicate from $A$ to $B$ can be positive only if $p>1/2$. 

To see that when $p>1/2$ there can indeed be positive capacity, simply let $W^{A\prec B}$ describe the identity channel on a subspace, and let $W^{A\succ B}$ induce a state $\sigma^{B_I}$ on the orthogonal subspace. Then \Cref{eq:aae} is effectively a quantum erasure channel, which has capacity $Q=\max\{0,2p-1\}$ that is positive for $p>1/2$.
\end{proof}

\begin{corollary}
$W^{AB}=pW^{A\prec B}+(1-p)W^{A\succ B}$ can have positive quantum communication capacity in at most one direction (either Alice to Bob or Bob to Alice).
\end{corollary}
\begin{proof}
By the previous theorem, to have positive capacity in %\sout{some}
\ablue{either} direction $p$ or $1-p$ has to be greater than $1/2$. Yet this can only hold for at \ablue{most} one of them.
\end{proof}

\begin{corollary}
$W^{AB}=\frac{1}{2}W^{A\prec B}+\frac{1}{2}W^{A\succ B}$ has no quantum communication capacity in either direction.
\end{corollary}
This is a simple consequence of the previous theorem. When the uncertainty in the causal order is maximal ($p=1-p=1/2$), there is no quantum communication capacity in either direction.%For the simple case where $W^{A\nsucceq B}$ is a channel from $A$ to $B$ and $W^{A\npreceq B}$ is a channel from $B$ to $A$, it is impossible to communicate quantum information in either way when $p=1/2$ (Theorem \ref{th:ngbwc}).

\section{Discussion}

%\item The conclusion of \Cref{th:ngbwc} does not hold for entanglement generation. Consider $W=(pC^{A\rightarrow B}+(1-p)C^{A\rightarrow B})\otimes \Phi^+=pC^{A\rightarrow B}\otimes \Phi^++(1-p)C^{A\rightarrow B}\otimes \Phi^+$. No matter what $p$ is, entanglement can be shared between $A$ and $B$ using the $\Phi^+$ part. \fcomment{Is this interesting? Entanglement generation doesn't seem like a relevant task, when the shared resource can be an entangled state.}
We have seen that some bias in the causal order is necessary to have any quantum communication, with the consequence that quantum information can only be exchanged in one direction when the same system is used for read-out and encoding. Interestingly, this is not the case for \emph{classical} communication: agents can transmit perfect classical bits asymptotically, as long as in each run there is a non-zero probability of having a channel in the right direction.
%\blue{This fact can help establish quantum communication even through a resource with random causal order, if the parties can share additional entanglement. Crucially, the entanglement has to be available at the encoding stage, before the parties can access the resource described by the process. As our result holds for mixtures of arbitrary channels with memory, it implies that shared entanglement cannot be used as a resource for quantum communication if it is only available together with a channel in a random direction.}

The impossibility to communicate  quantum information bidirectionally does not contradict recent results proving two-way quantum communication with a single particle, exchanged in a superposition of directions \cite{delsanto2018, massa2018experimental}. Indeed, in those scenarios each agent performs a preparation first and a measurement afterwards, with each agent's measurement always after the other's preparation. This corresponds to a four-partite process with fixed causal order, although with the interesting constraint that only one particle per run is exchanged.

%\sout{The conclusion that maximal uncertainty in the causal order forbids quantum communication is to be understood in a qualified context. The uncertainty is assumed to be classical, which leads to causally separable processes. The question about causally inseparable processes is left open. In addition,}
Finally, there are several promising directions to extend the analysis presented here. We have only considered classical uncertainty of causal order, modelled by causally separable processes. It remains to be established whether processes with \emph{indefinite} causal order \cite{oreshkov2012quantum, chiribella2013quantum} can outperform separable ones in this respect, for example if they permit bidirectional quantum communication. Since indefinite causal order can provide advantages in certain communication tasks \cite{feix2015quantum, guerin2016exponential, ebler2018enhanced}, it is an interesting open question whether it also constitutes a quantum communication resource in the asymptotic scenario treated here, in particular in view of the recent experimental interest \cite{procopio2015experimental,  rubino2017experimental, RubinoExperimentalOrders, Goswami2018, goswami2018communicating, wei2018experimental, guo2018experimental}. Furthermore, there are other communication settings such as different asymptotic settings or the one-shot setting where the results in this work do not apply. These and the general topic of quantifying causal order as a resource for communication are left for further investigations.

\section*{Acknowledgement}
We thank Alexander Smith for posing the question of the resource theory of causal structure to us, which triggered this work. %\fcomment{Should we acknowledge the Manitulin workshop? Did anybody circlate some template for that?}\dcomment{Feel free to edit the following directly without sout:
We thank the Spacetime and Information Workshop at Manitoulin Island organized and hosted by Natacha Altamirano, Fil Simovic, Alexander Smith, and Robert Mann, as well as the Mitacs Globalink program for fostering this research. DJ is grateful to Lucien Hardy and Achim Kempf for guidance and support.
Research at Perimeter Institute is supported by the Government of Canada through the Department of Innovation, Science and Economic Development Canada and by the Province of Ontario through the Ministry of Research, Innovation and Science. This work was supported by the Australian Research Council (ARC) Centre of Excellence for Quantum Engineered Systems grant (CE 110001013). This publication was made possible through the support of a grant from the John Templeton Foundation. The opinions expressed in this publication are those of the authors and do not necessarily reflect the views of the John Templeton Foundation. We acknowledge the traditional owners of the land on which the University of Queensland is situated, the Turrbal and Jagera people.

%\fcomment{The bibliography needs some fixing. Usually first names are not written in full, for example D. Ebler, L. M. Procopio, etc. Conventions vary regarding having a dot or not after the initial, but we should be consistent. One often uses ``et al'' when there are many authors; there are different convention, for example one can keep the first three authors when there are more than five in total. In PRL the titles are not written, although it might be nice to have the titles for the arxiv version. You can choose all those details, but try to make the choice uniform for all the biblio.}\dcomment{I've set the bibliographystyle to ``apsrev4-1'', which seems to be the one preferred by aps journals. Feel free to change it if you find some other one suitable.}

%\bibliographystyle{apsrev4-1}
\bibliography{mendeley.bib}
\end{document}